\documentclass[11pt]{article}
\usepackage[totalwidth=13.0cm,totalheight=20.0cm]{geometry}
\usepackage{latexsym,amsthm,amsmath,amssymb,url}
\usepackage[ruled, linesnumbered]{algorithm2e}
\usepackage{enumerate}
\usepackage{tikz}

\usetikzlibrary{decorations.pathreplacing}

\newtheorem{lemma}{Lemma}

\newtheorem{definition}{Definition}
\newtheorem{theorem}{Theorem}

\usepackage{authblk}
\title{Solving the $r$-pseudoforest Deletion Problem in Time Independent of $r$}
\author{Bin Sheng\thanks{Email: shengbinhello@nuaa.edu.cn}}
\affil{College of Computer Science and Technology, Nanjing University of Aeronautics and Astronautics, Collaborative Innovation Center of Novel Software Technology and Industrialization, Nanjing, Jiangsu, 211106, PR China}
\begin{document}

\maketitle

\begin{abstract}
The feedback vertex set problem is one of the most studied parameterized problems. Several generalizations of the problem have been studied where one is to delete vertices to obtain graphs close to acyclic.
In this paper, we give an FPT algorithm for the problem of deleting at most $k$ vertices to get an $r$-pseudoforest. A graph is an $r$-pseudoforest if we can delete at most $r$ edges from each component to get a forest. Philip et al. introduced this problem and gave an $O^*(c_{r}^{k})$ algorithm for it, where $c_r$ depends on $r$ double exponentially. In comparison, our algorithm runs in time $O^*((10k)^{k})$, independent of $r$.
\end{abstract}

\section{Preliminary}
The Feedback Vertex Set problem is one of the 21 problems proved to be NP-hard by Karp \cite{DBLPconf/coco}. It asks to delete a minimum number of vertices to make a given graph into a forest. The problem has important applications in  artificial intelligence \cite{DBLP:journals/siamcomp/Bar-YehudaGNR98,DBLP:journals/jair/BeckerBG00}, bio-computing \cite{DBLP:journals/tcbb/BlinSV10,DBLP:journals/jbcb/FuJ07}, operating system research \cite{DBLP:books/aw/SilberschatzG94}, and so on.

%The problem has attracted a lot of attention from the parameterized complexity community. Both its undirected and directed versions have been well studied \cite{DBLP:journals/algorithmica/CaoC015,DBLP:journals/jcss/ChenFLLV08,DBLP:journals/talg/ChitnisCHM15,DBLP:journals/siamdm/CyganPPW13}.
%Feedback vertex set problem is fixed-parameter tractable when parameterized by the solution size, the number of vertices to delete. For the undirected feedback vertex set problem, the state-of-the-art algorithm runs in time $O^*(3.619^k)$ in the deterministic setting \cite{DBLP:journals/algorithmica/CaoC015},
%and $O^*(3^k)$ in the randomized setting \cite{DBLP:conf/focs/CyganNPPRW11}. Here the $O^*$ notation hides polynomial factors in $n$, and $k$ is the solution size.

 The Feedback Vertex Set problem has attracted a lot of attention from the parameterized complexity community. Both its undirected and directed versions are fixed-parameter tractable when parameterized by the solution size \cite{DBLP:journals/algorithmica/CaoC015,DBLP:journals/jcss/ChenFLLV08,DBLP:journals/talg/ChitnisCHM15,DBLP:journals/siamdm/CyganPPW13}.
For the undirected case, the state-of-the-art algorithm runs in time $O^*(3.619^k)$ in the deterministic setting \cite{DBLP:journals/algorithmica/CaoC015},
and $O^*(2.7^k)$ in the randomized setting \cite{DBLP:conf/soda/LiN20}. Here the $O^*$ notation hides polynomial factors in $n$, and $k$ is the solution size.

Relaxing the acyclic requirement, researchers have defined several classes of almost acyclic graphs. A graph $F$ is an \textit{$r$-pseudoforest} if we can delete at most $r$ edges from each component in $F$ to get a forest. A \textit{pseudoforest} is a 1-pseudoforest. An \textit{almost $r$-forest} is a graph from which we can delete $r$ edges to get a forest.

%{As a generalization of feedback vertex set problem, }
Philip et al.\cite{DBLP:journals/siamdm/PhilipRS18} introduced the problem of deleting vertices to get an almost acyclic graph. There are several results in this line of research. In \cite{DBLP:journals/siamdm/PhilipRS18}, the authors gave an $O^*(c_{r}^{k})$ algorithm for $r$-pseudoforest deletion, which asks to delete at most $k$ vertices to get an $r$-pseudoforest. They also gave an $O^*(7.56^k)$ time algorithm for the problem of pseudoforest deletion. Bodlaender et al. \cite{DBLP:journals/dam/BodlaenderOO18} gave an improved algorithm for pseudoforest deletion running in time $O(3^knk^{O(1)})$.

Rai and Saurabh \cite{DBLP:journals/tcs/RaiS18} gave an $O^{*}(5.0024^{(k+r)})$ algorithm for the Almost Forest Deletion problem, which asks to delete at most $k$ vertices to get an almost $r$-forest. Lin et al.\cite{DBLP:journals/ipl/LinFWCFL18} gave an improved algorithm for this problem that runs in time $O^{*}(5^k4^r)$.

To the author's knowledge, there has been no follow-up research for the problem of $r$-pseudoforest deletion. In this paper, we provide a simple branching algorithm that runs in time $O^{*}((10k)^k)$. The $c_r$ in the algorithm by Philip et al. \cite{DBLP:journals/siamdm/PhilipRS18} is a large constant that depends on $r$ double exponentially. In contrast, the running time of our algorithm is independent of $r$, which is somewhat surprising. %Our algorithm improves over their result for 2-pseudoforest when the solution size is small.

\section{Notation and Terminology}
Here we give a brief list of the graph theory concepts used in this paper; for other notation and terminology, we refer readers to \cite{DBLP:books/daglib/0009415}.

For a graph $G=(V(G),E(G))$, $V(G)$ and $E(G)$ are its vertex set and edge set respectively.
A non-empty graph $G$ is \textit{connected} if there is a path between every pair of vertices. Otherwise, we call it \textit{disconnected}.  %A \textit{cut set} in a connected graph is a set of vertices whose deletion results in a disconnected graph. A connected graph $G$ is said to be \textit{$k$-connected} if every cut set of it contains at least $k$ vertices.  {A connected graph $G$ is \textit{$k$-edge-connected} if $G$ remains connected whenever less than $k$ edges are deleted from it.} An edge $e\in E(G)$ in a connected graph $G$ is called a \textit{bridge} if $G-e = (V(G), E(G)-e)$ is disconnected.  A \textit{block} in a graph is a maximal 2-connected subgraph.

The \textit{multiplicity} of an edge is the appearance number of it in the multigraph. An edge $uv$ is a \textit{loop} if $u=v$, and we call it a  loop at $u$. The \textit{degree} of a vertex is the number of its appearances as end-vertex of some edge. \textit{Bypassing} a vertex $v$ of degree 2 means to delete $v$ and add an edge between its two neighbors $u$ and $w$(even if $u=w$ or there is already an edge between $u$ and $w$). A \textit{forest} is a graph in which there is no cycle and a \textit{tree} is a connected forest.  A graph $H=(V(H),E(H))$ is a \textit{subgraph} of a graph $G=(V(G),E(G))$, if $V(H)\subseteq V(G)$ and $E(H)\subseteq E(G)$. A subgraph $H$ of $G$ is an \textit{induced subgraph} of $G$ if for every $u,v\in V(G)$, edge $uv\in E(H)$ if and only if $uv\in E(G)$. For $X\subseteq V(G)$, $G[X]$ denotes the subgraph of $G$ induced by $X$.

For a positive integer $r$, a graph is an \textit{r-pseudoforest} if we can make each component into a forest by deleting at most $r$ edges. %A \textit{pseudoforest} is a 1-pseudoforest.
%The \emph{contraction} of an edge $uv$ in $G$ removes $u$ and $v$ from $G$, and replaces them by a new vertex adjacent to exactly all the neighbours of $u$ and $v$ in $G$. {Note that, by its definition, edge contraction creates neither self-loops nor multiple edges.}

\section{FPT Algorithm for $r$-Pseudoforest Deletion}
 In this section, we give an FPT algorithm for $r$-pseudoforest deletion. Let us give the formal problem definition first.

 \quad

  { \textbf{$r$-Pseudoforest Deletion}} \nopagebreak

    \emph{Instance:} An undirected graph $G$, two integers $r$ and $k$.

    \emph{Parameter:} $k$.

    \emph{Output:} Decide if there is a set $X\in V(G)$ with $|X|\leq k$ such that $G-X$ is an $r$-pseudoforest.

\quad

Let $(G,k)$ be an instance of $r$-pseudoforest deletion, where $G$ contains $m$ edges and $n$ vertices.
%Without loss of generality, we may assume there is no connected component in $G$ which is a 2-pseudoforest as there is no need to delete any vertex from such a component.

%\begin{definition}
%A 2-pseudoforest with 2 vertices and 3 multiple edges is called a {\bf\textit small spindle}, see Figure \ref{fig:lemmaProof1} for an illustration.
%\end{definition}

%
%\begin{figure}
%\centering
%\begin{tikzpicture}[scale = 0.5]
%
%\draw [](5, 0) ellipse (2.5 and 1.2);
%
%\draw [](2.48,0)node[left]{$u$};
%
%\draw [](7.52,0)node[right]{$v$};
%
%\draw [](2.48,0)[fill]circle[radius=0.07];
%\draw [](7.52,0)[fill]circle[radius=0.07];
%
%\draw [](2.46,0)--(7.52,0);
%
%\end{tikzpicture}
%\caption{A small spindle.}
% \label{fig:lemmaProof1}
%\end{figure}

Now we give some reduction rules to simplify the given instance.

\textbf{Reduction Rule 1}\label{R:Reduction1}: If there is a component $C$ in $G$ that is an $r$-pseudoforest, then delete $C$ and get a new instance $(G-C,k)$.

Reduction Rule 1 is safe since there is no need to delete any vertex from a component that is an $r$-pseudoforest.

\textbf{Reduction Rule 2}\label{R:Reduction2}: If there are at least $r+1$ loops at a vertex $v$, then delete $v$ and decrease $k$ by 1.

Reduction Rule 2 is safe since the graph $G[v]$ is not an $r$-pseudoforest and every solution must contain $v$.

\textbf{Reduction Rule 3}\label{R:Reduction3}: If $G$ contains an edge $uv$ of multiplicity greater than $r+2$, then reduce its multiplicity to $r+2$.

Every solution intersects $\{u, v\}$ if edge $uv$ has multiplicity at least $r+2$. Reducing the multiplicity of $uv$ to $r+2$ does not affect the set of solutions. Thus, Reduction Rule 3 is safe.

\textbf{Reduction Rule 4}\label{R:Reduction4}: If $G$ contains a vertex $v$ of degree at most 1, then delete $v$ and get a new instance $(G-v,k)$.

On the one hand, every solution of $(G,k)$ is a solution of $(G-v,k)$. On the other hand, attaching a leaf to any component of an $r$-pseudoforest still gives an $r$-pseudoforest. And so every solution of $(G-v,k)$ is also a solution of $(G,k)$ since $d_G(v)\leq 1$. Thus, Reduction Rule 4 is safe.

\textbf{Reduction Rule 5}\label{R:Reduction6}: If $G$ contains a vertex $v$ of degree 2, then bypass $v$.

\begin{lemma}
Reduction Rule 5 is safe.%\textcolor{red}{needs revision}
\end{lemma}

\begin{proof}
Let $H$ be the graph obtained from $G$ by bypassing $v$. Let $N_G(v)=\{u, w\}$. We show that $(H,k)$ is a yes-instance if and only if $(G,k)$ is a yes-instance.

On the one hand, assume $(H,k)$ is a yes-instance, and $X$ is a solution of it. Suppose $X\cap \{u,w\}\neq \emptyset$, then $d_{G-X}(v)\leq 1$. By adding $v$ as an isolated vertex or a leaf to $H-X$, we get $G-X$. It follows that $G-X$ is an $r$-pseudoforest since $H-X$ is an $r$-pseudoforest. Thus, $X$ is also a solution of $(G,k)$.
 Otherwise, $X\cap \{u,w\}= \emptyset$. In this case, the edge $uw$ is in some component of the $r$-pseudoforest $H-X$. By subdividing $uw$ in $H-X$, we get $G-X$, which is still an $r$-pseudoforest. And so $X$ is a solution of $(G,k)$. In both cases, $(G,k)$ is a yes-instance, and $X$ is a solution of it.

On the other hand, assume $(G,k)$ is a yes-instance, and $X$ is a solution of it. Suppose $v\in X$. Consider $X'=X\cup\{u\}-\{v\}$. Note that $|X'|\leq |X|$. If $u\in X$, then $H-X'=G-X$. If $u\not\in X$, then $H-X'=G-X-\{u\}$. In both cases, $H-X'$ is an $r$-pseudoforest. Thus $X'$ is a solution for $(H,k)$, and $(H,k)$ is a yes-instance.

 Otherwise, $v\not\in X$. If $d_{G-X}(v)=2$, then we get $H-X$ from $G-X$ by bypassing $v$. If $d_{G-X}(v)= 1$, then we get $H-X$ from $G-X$ by deleting $v$. In both cases, the number of edges and vertices decrease by the same amount, thus $H-X$ is also an $r$-pseudoforest.  If $d_{G-X}(v)=0$, then we get $H-X$ from $G$ by deleting the isolated vertex $v$. Above all, we know that $(H,X)$ is a yes-instance, and $X$ is a solution of it.
\end{proof}

%{Observes that Reduction Rule 5 never creates any new vertex of degree 2 while Reduction Rule 6 decreases the number of vertices with degree 2. } Thus the algorithm  will not get into infinite loop of alternating applications between Reduction Rule 5 and 6. %In Reduction Rule 6, it is possible that we creates new loops, which happens only when the two neighbors of $v$ coincide. But, thus Reduction Rule 5 and 6 will not interleave with each other.}

\textbf{Reduction Rule 6}\label{Reduction7}:  If $k < 0$, terminate and conclude that $(G, k)$ is a no-instance.

We call $(G,k)$ a \textit{reduced instance} if none of Reduction Rules 1-6 applies to it.
Observe that if $(G,k)$ is a reduced instance, then $G$ satisfies the following conditions:
\,
\begin{enumerate}[(P1)]
\item each edge has multiplicity at most $r+2$;
\item the minimum degree is at least 3;
\item there are at most $r$ loops at any vertex.
\end{enumerate}

\begin{lemma}
If a connected graph $G$ is an $r$-pseudoforest, then $|E(G)| \leq |V(G)|+r-1$.
\end{lemma}

\begin{proof}
We can make each connected component of an $r$-pseudoforest into a forest by deleting at most $r$ edges. Since $G$ is a connected $r$-pseudoforest, $|E(G)| \leq |V(G)|+r-1$.
\end{proof}

%\begin{corollary}
%If a connected graph $G$ is a 2-pseudoforest, then $|E(G)| \leq |V(G)|+1$.
%\end{corollary}

%\begin{definition}\label{definition2}
%A 2-pseudoforest $F$ is called {\bf proper} if it satisfies the following conditions:
% \begin{enumerate}[(Q1)]
% \item there is no component with 1 vertex and at least 2 edges;
% \item there is no component with 2 vertices and at least 4 edges;
% \item all components with 2 vertices and 3 edges are small spindles.
% \end{enumerate}
%\end{definition}

\begin{definition}
An $r$-pseudoforest with one vertex and $r$ loops is called an \textbf{$r$-loop}. An $r$-pseudoforest with two vertices and $r+1$ edges is called an \textbf{$(r+1)$-edge}.
\end{definition}

\begin{lemma}\label{lemma3}
Let $(G,k)$ be a reduced instance of $r$-pseudoforest deletion. Let $X\subseteq V(G)$ be a subset of vertices such that $F=G-X$ is an $r$-pseudoforest. Suppose $F$ contains $t_1$ $r$-loops, $t'_1$ components consists of one vertex and at least $(r+2)/3$ loops, $t_2$ $(r+1)$-edges, $t'_2$ components consists of two vertices and at least $2(r+2)/3$ edges. If $|V(F)|= t_1+t'_1+2t_2+2t'_2 +s$, then the following statements hold.
\begin{enumerate}
\item $|E(F)|\leq rt_1+rt'_1+(r+1)t_2+(r+1)t'_2+(r+2)s/3$.
\item $2m\geq 3n+(2r+1)t_1+( 2(r+2)/3+1)t'_1+(r+2)t_2+ (2(r+2)/3+1 )t'_2$.
\item For any $c>0$, $\sum_{v\in X}(d(v)-c) \geq m-cn -[(r-c)t_1+(r-c)t'_1+( r+1-2c )t_2+( r+1-2c )t'_2+ (r+2-3c)s/3]$.
\end{enumerate}
\end{lemma}
\begin{proof}
%Proof of statement 1: this follows from the definition of $F$.

Proof of statement 1: Because $F$ is an $r$-pseudoforest, we know that $|E(C)|\leq |V(C)|+r-1$ for each component $C$ in $F$. Since the ratio $\frac{|V(C)|+r-1}{|V(C)|}$ decreases when $|V(C)|$ increases, we have $\frac{|E(C)|}{|V(C)|}\leq \frac{r+2}{3}$ for each component $C$ in $F$ with at least 3 vertices. It follows that $|E(F)|\leq rt_1+rt'_1+(r+1)t_2+(r+1)t'_2+(r+2)s/3$.

\quad

Proof of statement 2: Vertices in each $(r+1)$-edge have average degree $r+1$ and the vertex in an $r$-loop has degree $2r$. After exhaustive applications of Reduction Rule 1, there is no $r$-pseudoforest component in $G$.
 Note that there are $t_1$ $r$-loops, $t'_1$ components consists of one vertex and at least $(r+2)/3$ loops, $t_2$ $(r+1)$-edges, $t'_2$ components consists of two vertices and at least $2(r+2)/3$ edges in $F$. Thus $G$ contains at least $t_1$ vertices of degree at least $2r+1$, $t'_1$ vertices of degree at least $2(r+2)/3+1$ and $t_2$ vertices of degree at least $r+2$, $t'_2$ vertices of degree at least $2(r+2)/3+1$. Combining with the fact that $\delta(G)\geq 3$, we have  $$2m=\sum_{v\in V(G)} d(v) \geq 3n+(2r+1)t_1+( 2(r+2)/3+1)t'_1+(r+2)t_2+ (2(r+2)/3+1 )t'_2.$$

Proof of statement 3: Since every edge in $E(G)\setminus E(F)$ is incident to at least one vertex in $X$, we have $|E(G)| \leq |E(F)| + \sum_{v\in X}d(v).$ It follows that
\begin{equation}\label{inequality0}
m\leq rt_1+rt'_1+(r+1)t_2+(r+1)t'_2+(r+2)s/3 + \sum_{v\in X}d(v).
\end{equation}
 Since $|V(G)| = |V(F)| + |X|$, denoting $|X|$ by $x$, we have
 \begin{equation}\label{equality0}
n=t_1+t'_1 + 2t_2+2t'_2+s+x.
\end{equation}
Combining (\ref{inequality0}) and (\ref{equality0}), for any positive constant $c$, we have
%\begin{equation}\label{inequality1}
$\sum_{v\in X}(d(v)-c) \geq m-cn -[(r-c)t_1+(r-c)t'_1+( r+1-2c )t_2+( r+1-2c )t'_2+ (r+2-3c)s/3].$
\end{proof}

Let $\Phi = (v_1,v_2,\ldots,v_n)$ be an ordering of $V(G)$ that satisfies $d(v_1)\geq d(v_2)\geq\dots \geq d(v_n)$. Let $V_{10k}=\{v_i|1\leq i\leq 10k\}$ be the set of $10k$ vertices with the largest degrees. The following lemma shows that every solution of a reduced instance intersects $V_{10k}$.
\begin{lemma}\label{Lemma5}
Let $X$ be any solution of $r$-pseudoforest deletion for a reduced instance $(G,k)$, then $X\cap V_{10k}\neq \emptyset$.
\end{lemma}
\begin{proof}
Suppose there is a solution $X$ of $r$-pseudoforest deletion on $(G,k)$, which satisfies $X\cap V_{10k}= \emptyset$ and $|X|\leq k$. We show a contradiction by counting the number of edges in $G$.
%According to Lemma \ref{lemma4}, $F=G-X$ is a proper 2-pseudoforest.
Let $F=G-X$ be the resulting $r$-pseudoforest after deleting $X$ from $G$. Suppose $F$ contains $t_1$ $r$-loops, $t'_1$ components consisting of one vertex and at least $(r+2)/3$ loops, $t_2$ $(r+1)$-edges, and $t'_2$ components consisting of two vertices and at least $2(r+2)/3$ edges. Denote $|V(F)|=t_1+t'_1 + 2t_2+2t'_2+s$. In the following, we assume that $n=t_1+t'_1 + 2t_2+2t'_2+s+x \geq 51k$, as otherwise, we may solve the problem via brute force. The assumption also implies that $t_1+t'_1 + 2t_2+2t'_2+s \geq 50k$.

By the choice of $V_{10k}$, the degree of each vertex in $X$ is at most $d(v_{10k})$. Since $|X|\leq k$, it follows that

\begin{equation}\label{inequality2}
\sum_{i=1}^{10k}(d(v_i)-c)\geq 10\sum_{v\in X}(d(v)-c)%\geq 70(  m-cn -[(r-c)t_1+(r-c)t'_1+( r+1-2c )t_2+( r+1-2c )t'_2+ (r+2-3c)s/3]  ).
\end{equation}
%In the last inequality, we make use of inequality (\ref{inequality1}).

By the definition of $V_{10k}$ and the assumption that $X\cap V_{10k}= \emptyset$, we have $X\subseteq \{v_{10k+1},v_{10k+2},...,v_n\}$ and so
\begin{equation}\label{inequality3}
\sum_{i>10k}^{n}(d(v_i)-c)\geq \sum_{v\in X}(d(v)-c).
\end{equation}
Combining inequalities (\ref{inequality2}) and (\ref{inequality3}), we have
\begin{equation}
\sum_{v\in V(G)}(d(v)-c)\geq 11\sum_{v\in X}(d(v)-c).  %\geq 71(  m-cn -[(r-c)t_1+(r-c)t'_1+( r+1-2c )t_2+( r+1-2c )t'_2+ (r+2-3c)s/3]  ).
\end{equation}
Since $\sum_{v\in V(G)}d(v) = 2m$, it follows that
\begin{equation}
2m-cn \geq 11\sum_{v\in X}(d(v)-c).
\end{equation}
Thus
$2m-cn \geq $ $11\{m-cn -[(r-c)t_1+(r-c)t'_1+( r+1-2c )t_2+( r+1-2c )t'_2+ (r+2-3c)s/3]\}.$

%$10cn\geq 9m - 11[(r-c)t_1+(r-c)t'_1+( r+1-2c )t_2+( r+1-2c )t'_2+ (r+2-3c)s/3]   $

According to Statement 2 in Lemma 3, we have
$10cn\geq 4.5[3n+(2r+1)t_1+( 2(r+2)/3+1)t'_1+(r+2)t_2+ ( 2(r+2)/3+1 )t'_2] - 11[(r-c)t_1+(r-c)t'_1+( r+1-2c )t_2+( r+1-2c )t'_2+ (r+2-3c)s/3]. $

Therefore, $(10c-13.5)n\geq [4.5(2r+1)-11(r-c)]t_1+[(3(r+2)+4.5)-11(r-c) ]t'_1+[4.5(r+2)-11(r+1-2c)]t_2+[ (3(r+2)+4.5) - 11(r+1-2c)    ]t'_2-11(r+2-3c)s/3. $

Set $c=10r$ and simplify the above inequality, we get
%$(100r-13.5)(t_1+t'_1 + 2t_2+2t'_2+s+x)\geq t_1[9(2r+1)/2-11(r-10r)]+t'_1[9((r+2)/3+1/2)-11(r-10r)) ]+t_2[9(r+2)/2-11(r+1-20r)]+t'_2[ 9((r+2)/3+1/2) - 11(r+1-20r) ]-11(r+2-30r)s/3$
$(100r-13.5)(t_1+t'_1 + 2t_2+2t'_2+s+x)\geq (108r+4.5)t_1+(102r+10.5)t'_1+(213.5r-2)t_2+( 212r-0.5   )t'_2+11(29r-2)s/3.$

It follows that
$(100r-13.5)x \geq (18+8r)t_1 +(2r+24)t'_1+ ( 13.5r+25 )t_2+ (12r+26.5 )t'_2+(19r/3+37/6)s \geq (8t_1+2t'_1+ 13.5 t_2+ 12t'_2+19/3s)r \geq 2(n-x)r\geq 2(51-1) kr \geq 100kr$, a contradiction to the fact that $x\leq k$.

%which contradicts to the fact that $x\leq k$.

The above analysis shows that the assumption $X\cap V_{10k}=\emptyset$ is not correct, thus every solution of $(G,k)$ must intersect with $V_{10k}$.
\end{proof}

Lemma \ref{Lemma5} enables us to design the following algorithm for $r$-pseudoforest deletion.

  \begin{algorithm}[H]\label{algorithm}
  \caption{FPT algorithm for $r$-pseudoforest deletion}
  \KwIn{An undirected graph $G$, positive integers $r$ and $k$}
  \KwOut{yes, if $G$ has an $r$-pseudoforest deletion set of size at most $k$; no, otherwise.}
        If $G$ is an $r$-pseudoforest, then return yes. Else if $k\leq 0$, then return no.

        Exhaustively apply Reduction Rules 1-6 on $(G,k)$. Either get a reduced instance $(G',k')$ or return no.

        If $|V(G')|\leq 51k$, then apply a trivial branching algorithm to solve $(G',k')$ in time ${2}^{51k}n^{O(1)}$.

        Otherwise, $|V(G')|> 51k$. Order the vertices in $V(G')$ in non-increasing order according to the vertex degree. Let $V'_{10k}$ be the set of $10k$ vertices with the largest degrees. For each vertex $u\in V'_{10k}$, solve $(G'-u,k-1)$ recursively.

        If any branch gives a solution $X'_u$ of $r$-pseudoforest deletion on $(G'-u,k-1)$, then $X'_u\cup \{u\}$ is a solution on $(G',k)$. Otherwise, return no.
    \end{algorithm}

\begin{theorem}
The $r$-pseudoforest deletion problem parameterized by the solution size $k$ can be solved in time $O^*((10k)^k)$.
\end{theorem}

\begin{proof}
Let $(G,k)$ be a given instance of the $r$-pseudoforest deletion problem, where $G$ is an undirected graph and $k$ is an integer. Each recursive call in Step 4 of Algorithm \ref{algorithm} decreases the parameter by 1, and thus the height of the search tree is at most $k'$. At each step, the problem branches into at most $10k'$ subproblems. Hence, the number of leaves in the search tree is at most $(10k')^{k'} \leq (10k)^k$. It follows that Algorithm 1 solves the problem of $r$-pseudoforest deletion and runs in time $O^*((10k)^k)$.
\end{proof}
\section{Conclusion}
In this paper, we design an FPT algorithm for the $r$-pseudoforest deletion problem. The running time of our algorithm is independent of $r$, and so our algorithm improves over the result in \cite{DBLP:journals/siamdm/PhilipRS18}, when $r$ is large compared with $k$.
%Our algorithm relies on the fact that 2-pseudoforest is sparse. For any graph, its distance to a 2-pseudoforest is at most the size of its minimum feedback vertex set. It would be interesting to know the parameterized tractability of other problems when parameterized by this measure.
\section{Acknowledgements}
This research is supported by the National Natural Science Foundation of China (No. 61802178).
\bibliographystyle{plain}
%\cleardoublepage
\bibliography{references}
\end{document}